\documentclass[10pt,conference]{IEEEtran}
\ifCLASSINFOpdf
\else
\fi

\usepackage{array}
\usepackage{subfig}
\usepackage{amsmath}
\usepackage[nolist,nohyperlinks]{acronym}
\usepackage{graphicx,wrapfig,lipsum}
\usepackage{rotating}
\usepackage{amsmath, amssymb, amsthm}
\usepackage{stmaryrd}
\usepackage{tikz}
\usepackage{verbatim}
\usepackage{url}
\usepackage[utf8]{inputenc}
\usepackage[english]{babel}
\newtheorem{theorem}{Theorem}[]

\newtheorem{lemma}[]{Lemma}
\usepackage{multicol}
\usepackage{xr-hyper}
\usepackage{algorithm}
\usepackage{algpseudocode}
\usepackage{multicol}
\usepackage{setspace}
\newtheorem{definition}{Definition}
\usepackage[noadjust]{cite}
\usepackage{bbm}
\usepackage[normalem]{ulem}
\usepackage{color}
\usepackage{dsfont}
\usepackage[margin=0.70in]{geometry}
\newcommand\blfootnote[1]{%
  \begingroup
  \renewcommand\thefootnote{}\footnote{#1}%
  \addtocounter{footnote}{-1}%
  \endgroup
}
\usepackage{flushend}
\begin{document}
	%
	\title{Modeling and Analysis of HetNets with mm-Wave Multi-RAT Small Cells Deployed Along Roads }
	\author{Gourab Ghatak$^{\dagger}$ $^\ddagger$, Antonio De Domenico$^{\dagger}$, and Marceau Coupechoux$^\ddagger$
 \\ { $^{\dagger}$CEA, LETI, MINATEC, F-38054 Grenoble,
France; $^\ddagger$LTCI, Telecom ParisTech, Universit\'e Paris Saclay, France.}
\\ {Email: gourab.ghatak@cea.fr; antonio.de-domenico@cea.fr, and marceau.coupechoux@telecom-paristech.fr}}
		\maketitle
        \blfootnote{The research leading to these results are jointly funded by the European Commission (EC) H2020 and the Ministry of Internal affairs and Communications (MIC) in Japan under grant agreements Nº 723171 5G MiEdge.}
        \thispagestyle{empty}
	\begin{abstract}
We characterize a multi tier network with classical macro cells, and multi radio access technology (RAT) small cells, which are able to operate in microwave and millimeter-wave (mm-wave) bands. The small cells are assumed to be deployed along roads modeled as a Poisson line process. This characterization is more realistic as compared to the classical Poisson point processes typically used in literature. In this context, we derive the association and RAT selection probabilities of the typical user under various system parameters such as the small cell deployment density and mm-wave antenna gain, and with varying street densities. Finally, we calculate the signal to interference plus noise ratio (SINR) coverage probability for the typical user considering a tractable dominant interference based model for mm-wave interference. Our analysis reveals the need of deploying more small cells per street in cities with more streets to maintain coverage, and highlights that mm-wave RAT in small cells can help to improve the SINR performance of the users.
	\end{abstract}
	\IEEEpeerreviewmaketitle
    \section{Introduction}
To meet the tremendous increase in demand of high data rates in future wireless networks, the use of mm-wave bands is an attractive solution. However, mm-wave transmissions are associated with high path-loss and sensitivity to blockages~\cite{rappaport2013millimeter}. Therefore, to maintain ubiquitous coverage, mm-wave technology will be overlayed on top of the existing classical $\mu$-wave architecture. In an urban scenario, these mm-wave base stations are envisaged to be deployed along the roads e.g. on top of buildings and lamp-posts to cater to the needs of outdoor users.

In the context of heterogeneous networks, the user performance is often analyzed with the help of stochastic geometry, i.e., in terms of signal to interference plus noise (SINR) coverage probability and rate coverage probability~\cite{andrews2011tractable}. These metrics have been derived to investigate single-tier~\cite{bai2015coverage} and multi-tier mm-wave networks~\cite{di2015stochastic}. Elshaer et al.~\cite{elshaer2016downlink} have analyzed a multi-tier network with $\mu$-wave macro cells and mm-wave small cells in terms of user association, SINR and rate coverage, in both uplink and downlink scenarios. 
However, in these works, the base station locations are modeled as classical homogeneous Poisson point processes on the $\mathbb{R}^2$ plane~\cite{lee2013stochastic}, or as Poisson cluster processes~\cite{chun2015modeling}, which are not realistic representations of the network architecture in an urban scenario. 

To address this issue, we investigate a network geometry, where the small cells are deployed along the roads. In this regard, we take help of a framework introduced by Morlot~\cite{6260478} based on a Poisson line tessellation to model the roads in an urban scenario. Furthermore, we consider that the small cells are equipped with multi-radio access technology (RAT), thereby enabling them to opportunistically serve the users with both micro- and mm-wave bands.

The contribution of this paper is summarized as follows. We characterize a novel multi-tier network with small cells deployed along the streets, and derive the association probabilities of the typical user. Then, we consider a dominant interferer based model to characterize the mm-wave interference. On one hand, this approach of modeling mm-wave interference is more tractable than to consider all interfering base stations, whereas on the other hand, we show that it is more accurate in characterizing the SINR coverage as compared to a noise-limited approach~\cite{7593259}. Using these results, we derive the SINR coverage probability of the typical user, and investigate the effect of different deployment parameters of the network on the SINR performance. Our analysis reveals the fact that in cities with more streets, the operator must necessarily deploy more small cells per street to maintain the SINR coverage. Moreover, we highlight that the utility of multi-RAT base stations is not only limited to providing high data rate access to the users, but also that this technology, by taking advantage of the directional antennas, can considerably improve the SINR. Finally, we show that this gain in SINR performance brought by mm-wave, reaches a maximum value for a certain small cell deployment density, depending on the street density, and saturates at denser deployments.

The rest of the paper is organized as follows: in Section \ref{Sec:SM} we introduce the network architecture. We derive some preliminary results related to the mm-wave interference model, and the network geometry, in Section \ref{Sec:Prel}. In Section \ref{Sec:AP} and \ref{Sec:SCP} we compute the association probabilities and SINR coverage probability of a typical user, respectively. In Section \ref{Sec:NRD} we present some numerical results to discuss salient trends of the network. Finally, the paper concludes in Section \ref{Sec:Con}.
	\section{System Model}
    \label{Sec:SM}
    We consider a multi-tier cellular network consisting of macro base stations (MBSs) and small cell base stations (SBSs). The MBSs are deployed to ensure continuous coverage to the users. Whereas, the multi-RAT SBSs, deployed along the roads, locally provide high data rate by jointly exploiting $\mu$-wave and mm-wave bands. We assume that the same $\mu$-wave band is shared by MBSs and SBSs. 
From the perspective of the users, the base stations can either be in line-of-sight (LOS), or non line-of-sight (NLOS). In our analysis, we use the subscript notation $t,v,r$ to characterize the base stations, where $t \in \{M,S\}$ denotes the tier (MBS or SBS), $v \in \{L,N\}$ denotes the visibility state (LOS or NLOS), and $r \in \{\mu,m\}$ denotes the RAT ($\mu$-wave or mm-wave).
    \subsection{Network Model}
The MBS locations are modeled as a homogeneous Poisson point process (PPP) $\phi_M$ with intensity $\lambda_M$. On the contrary, the roads are modeled as Poisson line processes (PLP) with intensity $\lambda_R$. The SBSs are deployed on the PLP tessellation of the roads, according to a PPP $\phi_S$ with intensity $\lambda_S$. We denote by $\phi_i$, the 1D PPP on each road, where $i$ is the index of the roads. Furthermore, we consider outdoor users, which are modeled as stationary PPP $\phi_{OU}$ along the roads, with an intensity $\lambda_{OU}$. Thus, both the SBSs and users are modeled by a Cox process driven by the intensity measure of the road process~\cite{chiu2013stochastic}. In the following, we carry out our analysis from the perspective of the typical user~\cite{chiu2013stochastic}, located at the origin. 

\subsection{Blockage and Path-loss}
We assume a LOS ball model to categorize the MBSs into either LOS or NLOS processes, from the perspective of a user: {$\phi_{ML}$ and $\phi_{MN}$, respectively. As per the LOS ball approximation introduced in \cite{bai2015coverage}, let $D_M$ be the MBS LOS ball radius. The probability of the typical user to be in LOS from a MBS at a distance $r$ is $p_M(r)=1$, if $ r < D_M$, and $p_M(r)=0$, otherwise. All the SBSs lying on the same road as that of the typical user are considered to be in LOS, denoted by the process $\phi_{SL}$. All the other SBSs, are considered to be in NLOS, denoted by the process $\phi_{SN}$.

We assume that the path-loss at a distance $d_{tvr}$ from a transmitter is given by: $l_{tvr}(d) = K_{tvr}d_{tvr}^{-\alpha_{tvr}}$. $K$ and $\alpha$ are the path-loss coefficient and exponent, respectively. For $\mu$-wave communications, we assume a fast fading that is Rayleigh distributed with variance equal to one. On the contrary, due to the low local scattering in mm-wave~\cite{rappaport2013millimeter}, we consider a Nakagami fading for mm-wave communications~\cite{7593259}.
 Moreover let $G_0$ be the directional antenna gain in mm-wave transmissions. Thus, the average received power is given by $P_{tvr} = P_tK_{tvr}d_{tvr}^{-\alpha_{tvr}}$, in $\mu$-wave and $P_{tvr} = G_0P_tK_{tvr}d_{tvr}^{-\alpha_{tvr}}$ in mm-wave; where $P_t$ is the transmit power of a BS of tier $t$.
    \section{Preliminaries}
    \label{Sec:Prel}
     \subsection{Interference in LOS SBS mm-Wave Operation}
    We assume that in mm-wave operations, a user experiences interference only from the neighboring mm-wave SBS, due to the highly directional antenna. In Section~\ref{Sec:NRD}, we prove the accuracy of this assumption with Monte Carlo simulations. In this section, we model the probability that the typical user experiences interference from the neighboring SBS.   
    \begin{definition}
       We define 'spillover' as the region of interference generated by a mm-wave SBS to the coverage area of a neighboring SBS, while serving a user near its cell edge.
    \end{definition} 
        \begin{figure} 
\centering
\includegraphics[width=8cm,height = 3.5cm]{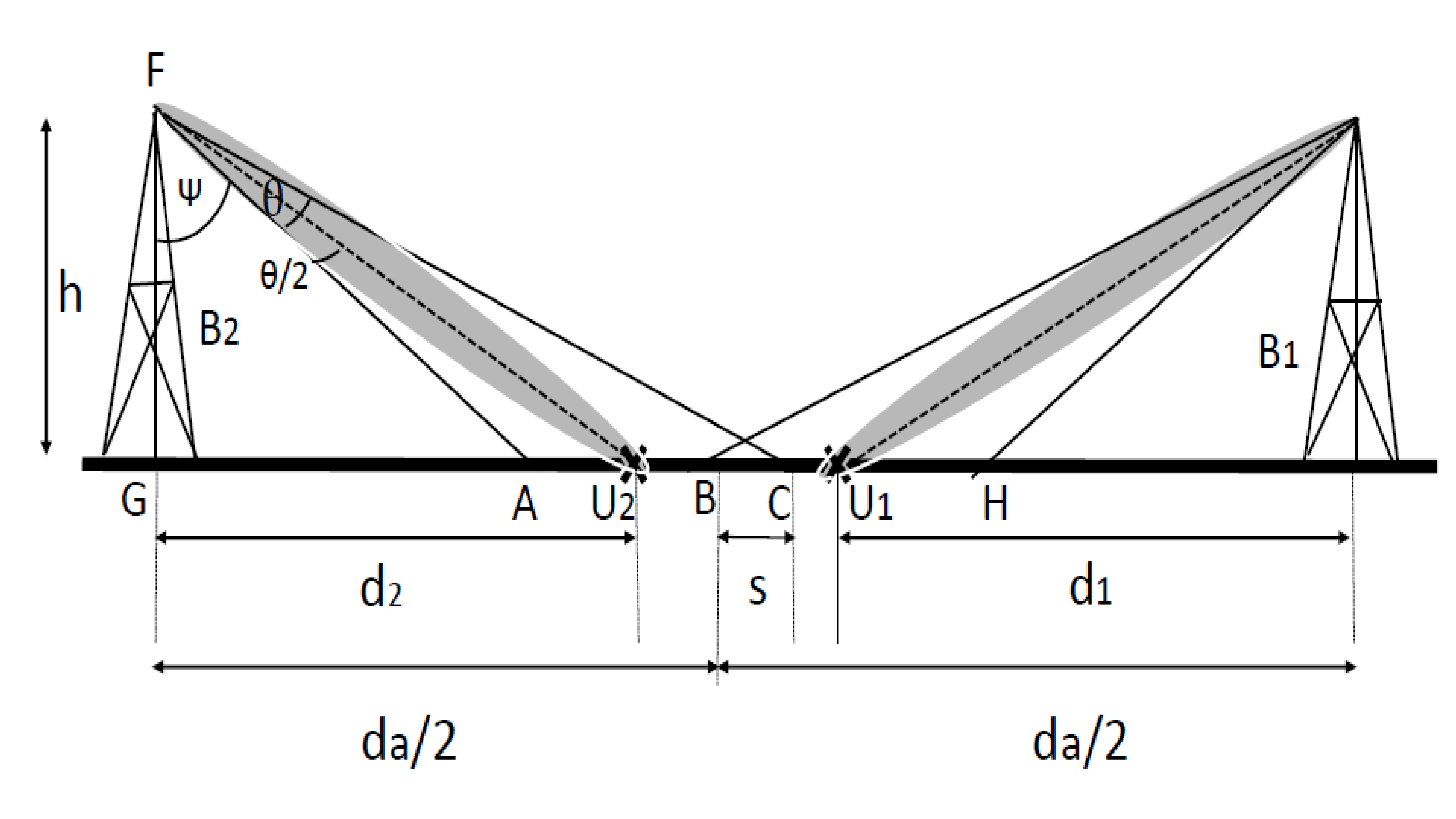}
\caption{Interference in mm-wave operation.}
\label{fig:mmFig}
\end{figure}
   \begin{lemma}
   For a typical user being served with mm-wave, the probability of experiencing mm-wave interference ($p_G$) from its closest neighboring SBS is given by \eqref{eq:mm_inter},
   \begin{figure*}
   \begin{align}
    p_G =  \int_{d^*}^{\hat{d}}\int_{d'}^{\frac{x}{2}} \exp\left(-\mu_{S}\left(x - h\tan\left(\frac{\theta}{2} + \tan^{-1}\frac{y}{h}\right)\right)\right) 
   \left(1-\exp\left(\mu_{OU}\left(\frac{x}{2}-d'\right)\right) \right) f_{yx}(y,x) dy dx
    \label{eq:mm_inter}
\end{align}
    \hrulefill
   \end{figure*}   
   where $\theta$ is the beam-width of the directional antenna, $d' = h\tan\left(\tan^{-1}\frac{x}{2h} - \frac{\theta}{2}\right)$, $d^* = \max \left(\frac{h - \sqrt{h^2 - 8h^2\tan\left(\frac{\theta}{2}\right)}}{2\tan\left(\frac{\theta}{2}\right)},2h\tan\left(\frac{\theta}{2}\right)\right)$, $\hat{d} = \frac{h + \sqrt{h^2 - 8h^2\tan\left(\frac{\theta}{2}\right)}}{2\tan\left(\frac{\theta}{2}\right)}$ and $f_{xy}(x,y) = 2\lambda_S^2\exp(-\lambda_S(x))$.
   \label{lem:p_G}
   \end{lemma}
    \begin{proof}
   See Appendix \ref{App:p_G}.
    \end{proof}   
 \subsection{Characterization of the NLOS SBS Cox Process}   
\begin{lemma}
    The pdf of the distance from a typical user to the nearest NLOS SBS is given by \eqref{eq:pdf_ds1}.
    \begin{figure*}
     \begin{align}
     f_{d_{S1}}(x) = 2\pi\lambda_R\exp\left(-2\pi\lambda_R \left(x+\int_{0}^x\exp\left(-2\lambda_S\sqrt{x^2 - r^2}\right)dr\right)\right)   \left[\lambda_S x\int_0^x \frac{\exp(-2\lambda_S\sqrt{x^2-r^2})}{\sqrt{x^2 - r^2}}dr\right]
     \label{eq:pdf_ds1}
    \end{align}
    \hrulefill
    \end{figure*}
    \label{lem:nearpt}
    \end{lemma}    
 \begin{proof}
  See Appendix \ref{App:nearpt}.
   \end{proof}      
      \begin{lemma} (\cite{6260478}, Theorem III.1).
    The SBS process $\phi_S$ is stationary and isotropic, with intensity $\pi \lambda_R \lambda_S$. Under Palm, it is the sum of $\phi_S$, of an independent Poisson point process on a line through the origin O with a uniform independent angle, and of an atom at O.
    \label{lem:COX}
    \end{lemma} 
   \begin{lemma}
   	The probability generating functional (PGF), for a class of radially symmetric functions $\nu$, of the Poisson Line Cox Process $\phi_S$ is given by \eqref{eq:PGF_COX}.
    \label{lem:Cox_PGF}
   \end{lemma}
    \begin{proof}
      See Appendix \ref{App:Cox}.
    \end{proof}
    \begin{figure*}
    \begin{align}
    G_{\phi_S}(\nu) = \exp\left(-2\pi\lambda_R\left(\int_{0}^\infty 1 - \exp\left(-2\lambda_S\int_{0}^\infty 1 - \nu\left(\sqrt{r^2 + t^2}\right)dt\right)\right)dr\right)
    \label{eq:PGF_COX}
    \end{align}
    \hrulefill
    \end{figure*}
    \begin{lemma}
    The PGF for a class of radially symmetric functions $\nu$, of a PPP on a randomly oriented line, passing through a point at a distance $d$ from the origin, is given by \eqref{eq:PGF_line}.
    \begin{figure*}
    \begin{align}
    G_{\phi_i,d}(\nu) = \frac{1}{2\pi}\int_0^{2\pi} \exp\left(-2\lambda_S\int_{0}^{\infty} \left(1- \nu\left(\left(d^2 + t^2 + 2td \cos\theta\right)^{\frac{1}{2}}\right)\right)dt\right)d\theta
    \label{eq:PGF_line}
    \end{align}
    \hrulefill
    \end{figure*}
    \end{lemma}
 \begin{proof}
 Without loss of generality, we assume that the line passes through $(d,0)$ inclined at an angle $\theta$ with the $x$-axis. Thus a point on the line at a distance $t$ from $(d,0)$ is at a distance 
 $
  r = \sqrt{(d + t\cos\theta)^2 + (t\sin \theta)^2},
 $
  from the origin. Taking the PGF along all such points completes the proof.
  \end{proof} 
    \section{Association Probabilities}
    \label{Sec:AP}
   	We assume that the BSs send their control signals in the $\mu$-wave band, due to the higher reliability of $\mu$-wave signals as compared to the mm-wave signals~\cite{shokri2015millimeter}.
For the association, a user compares the $\mu$-wave signals from the strongest LOS and NLOS SBS and MBS. According to our MBS LOS ball assumption, the received power from and LOS MBS is always greater than that received from and NLOS MBS. Accordingly, for association, we consider an NLOS MBS if and only if an LOS MBS is absent.
In case the user is associated with an MBS or an NLOS SBS, it is served in the $\mu$-wave band. Whereas, in case it is associated to an LOS SBS, the user compares the power received in the $\mu$-wave and mm-wave band, and selects the RAT providing the highest power.
\subsection{Tier Selection Probabilities}
In the following analysis, we drop the subscript $\mu$ for ease of notation. The term '$1$' in the subscript refers to the strongest BS of type $tv$. Accordingly, $d_{tv1}$ denotes the distance corresponding to the strongest base station of tier $tv$. Let the pdf of $d_{tv1}$ be denoted by $f_{tv1}(x)$. For, $\{t,v\} = \{SN\}$, $f_{tv1}(x)$ is given by \eqref{eq:pdf_ds1}. Whereas, for $\{t,v\} \neq \{SN\}$, the expressions for $f_{tv1}$, can easily be obtained by differentiating the void probabilities of the corresponding processes~\cite{chiu2013stochastic} : 
\begin{align}
f_{SL1}(x) &= 2\lambda_S \exp\left(-2\lambda_S x\right) \nonumber \\
f_{ML1}(x) &= 2\pi\lambda_Mx\exp\left(-\pi\lambda_Mx^2\right); \quad x < D_M \nonumber \\
f_{MN1}(x) &= 2\pi\lambda_Mx\exp\left(-\pi\lambda_M\left(x^2-D_M^2\right)\right); \hspace*{0.1cm} x \geq D_M \nonumber
\end{align}

\begin{lemma}
The tier selection probability of a user with a LOS and NLOS MBS and LOS SBS is given by \eqref{eq:Asso_MBS},
\begin{figure*}
\begin{align}
\mathbb{P}_{ML} &= 2\lambda_SW_1\mathbb{E}_{d_{SN1}}\left[1 - \exp\left(-\pi\lambda_M\left(\frac{P_S}{P_M}\right)\right)^{-\frac{2}{\alpha_{ML}}} d_{S1}^{\frac{2\alpha_{SN\mu}}{\alpha_{ML}}}\right]\int_0^\infty \left(W_{ML} \exp(-2\lambda_S x)\right)dx, \nonumber \\
\mathbb{P}_{MN} &= 2\lambda_S(1-W_1)\mathbb{E}_{d_{SN1}}\left[1 - \exp\left(-\pi\lambda_M\left(\frac{P_S}{P_M}\right)\right)^{-\frac{2}{\alpha_{MN}}} d_{S1}^{\frac{2\alpha_{SN\mu}}{\alpha_{MN}}}\right]\int_0^\infty \left(W_{MN}\exp(-2\lambda_S x)\right)dx, \label{eq:Asso_MBS}\\
\mathbb{P}_{SL} &= 2\lambda_S W_2\left( W_1 \int_0^\infty \left(1-W_{ML}\right) \exp(-2\lambda_S x)dx + \left(\int_0^\infty \left(1-W_{MN}\right)\exp(-2\lambda_S x)dx\right) \left(1 - W_1\right)\right) \nonumber
\end{align}
\hrulefill
\end{figure*}
where,
\begin{align}
W_{ML} &= 1 - \exp\left(-\pi\lambda_M\left(\frac{P_S}{P_M}\right)^{-\frac{2}{\alpha_{ML}}}x^{\frac{2\alpha_{SL\mu}}{\alpha_{ML}}}\right),\nonumber \\
W_{MN} &= 1 - \exp\left(-\pi\lambda_M\left(\frac{P_S}{P_M}\right)^{-\frac{2}{\alpha_{MN}}}x^{\frac{2\alpha_{SL\mu}}{\alpha_{MN}}}\right),\nonumber \\ 
W_1 &= \mathbb{E}[\mathds{1}(ML)] = 1 - \exp(-\pi\lambda_MD_M^2),\nonumber \\
W_2  &= \mathbb{E}_{d_{SN1}}\left[1-\exp\left(-2\mu d_{SN1}^{\frac{\alpha_{SN\mu}}{\alpha_{SL\mu}}}\right)\right].\nonumber
\end{align}
Here, $\mathds{1}(.)$ is the indicator function, and accordingly, $\mathbb{E}[\mathds{1}(ML)]$ denotes the probability that at least one LOS MBS exists.
\label{lem:asso}
\end{lemma}
\begin{proof}
See Appendix \ref{App:asso}.
\end{proof}
\begin{lemma} Given that a user is associated to a tier $t$ of visibility state $v$, the probability density function (pdf) of the distance of the serving BS is given by:
        \begin{equation}
\hat{f}_{ tv1}(x) = \frac{{f}_{ tv1}(x)}{\mathbb{P}_{tv}}\prod_{\forall (t'v' \neq tv)}\mathbb{P}(\phi_{t'v'} \cap (0,x) = 0),
\label{eq:cond_dist}
\end{equation}
\label{Lem:TagBS}
\end{lemma}
\subsection{RAT Selection Probability}
In case of LOS SBS association, the user selects $\mu$-wave or mm-wave RAT by comparing the received power from the selected SBS in these two bands. 
\begin{lemma}
The conditional mm-wave selection probability, given that it is associated with an LOS SBS is given by:
\begin{align}
\mathbb{P}_m = \exp\left(-2\lambda_S \left(\frac{K_\mu}{K_mG_0}\right)^{\frac{1}{\alpha_{SL\mu} - \alpha_{SLm}}} \right) \nonumber
\end{align}
\end{lemma}
 \begin{proof}
 We have :
\begin{align}
\mathbb{P}_m &= \mathbb{P}(r = mm|t = SL) \nonumber \\
&= \mathbb{P}(K_mG_0 P_S d_{SL1}^{-\alpha_{SLm}} > K_\mu P_S d_{SL1}^{-\alpha_{SL\mu}}) \nonumber \\
& = \mathbb{P}\left(d_{SL} > \left(\frac{K_\mu}{K_mG_0}\right)^{\frac{1}{\alpha_{SL\mu} - \alpha_{SLm}}}\right). \nonumber
\end{align}
Taking the void probability completes the proof.
\end{proof}
The overall association probability of the typical user is given by $\mathbb{P}_{tvr} = \mathbb{P}_{tv}\mathbb{P}_m$ where, the term $\mathbb{P}_m$ is considered only in case of association with a base station of type $SL$. In case of other tiers, we have exclusively, $r = \mu$.
\section{SINR Coverage Probabilities}
According to the derived association probabilities, the SINR coverage probability is obtained as:
\label{Sec:SCP}
\begin{theorem}
The conditional SINR coverage probability, given that the typical user is associated to a BS to type 'tv' in $\mu$-wave and mm-wave are given by \eqref{eq:SINR1} and \eqref{eq:SINR2}, respectively,
\begin{figure*}
\begin{align}
\label{eq:SINR1}
\mathbb{P}\left(SINR_{tv\mu} \geq \gamma \right) = 
\begin{cases} 
\mathbb{E}\left[\exp\left(-\frac{\gamma \sigma^2_{\mu}}{P_S K_\mu  d_{tv1}^{-\alpha_{SL\mu}}}\right)\right] \cdot  \prod\limits_{\substack{\{t'v'\}\\\neq \{tv\}}} \mathbb{E}_{d_{tv1}}\left[G_{\phi_{t'v'}}\left(\frac{P_{t}||x||^{\alpha_{t'v'}}}{P_{t}||x||^{\alpha_{t'v'}} + \gamma P_t'd_{tv1}^{\alpha_{tv}}}\right)\right] \cdot \\ 
\cdot \mathbb{E}_{d_{tv1}}\left[G^{tv1}_{\phi_{tv}}\left(\frac{||x||^{\alpha_{tv}}}{||x||^{\alpha_{tv}} + \gamma d_{tv1}^{\alpha_{tv}}}\right)\right] ; \qquad \forall \{tv\} \neq \{SN\} 
\\
 \mathbb{E}\left[\exp\left(-\frac{\gamma \sigma^2_{\mu}}{P_S K_\mu  d_{SN1}^{-\alpha_{SL\mu}}}\right)\right]  \cdot \prod\limits_{\substack{\{t'v'\}\\\neq \{SN\}}} \mathbb{E}_{d_{SN1}}\left[G_{\phi_{t'v'}}\left(\frac{P_{S}||x||^{\alpha_{t'v'}}}{P_{S}||x||^{\alpha_{t'v'}} + \gamma P_t'd_{SN1}^{\alpha_{SN}}}\right)\right] \cdot  \\ \cdot
\mathbb{E}_{d_{SN1}}\left[G^{SN1}_{\phi_{SN}}\left(\frac{||x||^{\alpha_{SN}}}{||x||^{\alpha_{SN}} + \gamma d_{SN1}^{\alpha_{SN}}}\right)\right]\cdot \mathbb{E}_{d_{SN1}}\left[G^{SN1}_{\phi_{i},d_{SN1}}\left(\frac{||x||^{\alpha_{SN}}}{||x||^{\alpha_{SN}} + \gamma d_{SN1}^{\alpha_SN}}\right)\right]; \quad \mbox{otherwise}.
\end{cases}
\end{align}
\begin{align}
\mathbb{P}\left(SINR_{SLm} \geq \gamma \right) =  \sum_{n = 1}^{n_0}\left(-1\right)^{n+1} \binom {n_0}n\mathbb{E}_{d_{SL1}}&\left[\exp\left(- \frac{n\gamma \sigma_{mm}^2}{K_mP_Sd_{SL1}^{-\alpha_{SLm}}G_0}\right)\right]  \mathbb{E}\left[\left( \frac{d_{SL2}^{\alpha_{SLm}}}{d_{SL2}^{\alpha_{SLm}} + \gamma p_G d_{SL1}^{\alpha_{SLm}}}\right)\right]\label{eq:SINR2}
\end{align}
\hrulefill
\end{figure*}
where, the expectations with respect to the serving BS distance $d_{tv1}$ is taken as per Lemma \ref{Lem:TagBS}. $G_{\phi}$ and $G_{\phi}^{y}$ refer to the PGF w.r.t. the process $\phi$, and the PGF w.r.t. $\phi$ taken according to the reduced Palm distribution with the first point at $y$, respectively.
\label{theo:SINR}
\end{theorem}
\begin{proof}
See Appendix \ref{App:SINR}.
\end{proof}
Finally, the overall coverage probability is calculated as:
\begin{equation}
	\mathbb{P}_C(\gamma) = \!\!\!\!\!\!\!\!\!\!\!\!\sum\limits_{t\in \{M, S\},\;v\in \{L,N\},\;r\in \{\mu, m\}}\!\!\!\!\!\!\!\!\!\!\!\!\mathbb{P}(SINR_{t,v,r}>\gamma|t,v,r)\mathbb{P}_{tvr}, \nonumber
    \label{eq:CovProb}
	\end{equation}
    where $r = m$ is considered only in case of $\{tv\} = \{SL\}$.
    \section{Numerical Results and Discussion}
    In this section, we provide some numerical results to discuss the salient trends of the network. We assume transmit powers of $P_M = 45$ dBm and $P_S = 30$ dBm. Parameters $K_{tvr}$ are derived from 3GPP UMa model for $\mu$-wave MBSs, Umi model for $\mu$-wave SBSs \cite{36.814}, and Umi model for mm-wave data transmission in SBSs~\cite{38.900}. The path-loss exponents are assumed to be $\alpha_{tNr} = 4$ and $\alpha_{tLr} = 2$ for the NLOS and LOS base stations. Furthermore, we assume a bandwidth of 20 MHz and 1 GHz for $\mu$-wave and mm-wave, respectively. The LOS ball for the macro tier is assumed to be 200 m and the MBS density is assumed to be $\lambda_S = 1$ km$^{-2}$.
    \subsection{Simplifying Approximations and Validation of the Model}
     \begin{figure}[t]
\centering
\includegraphics[width=8cm,height = 4cm]{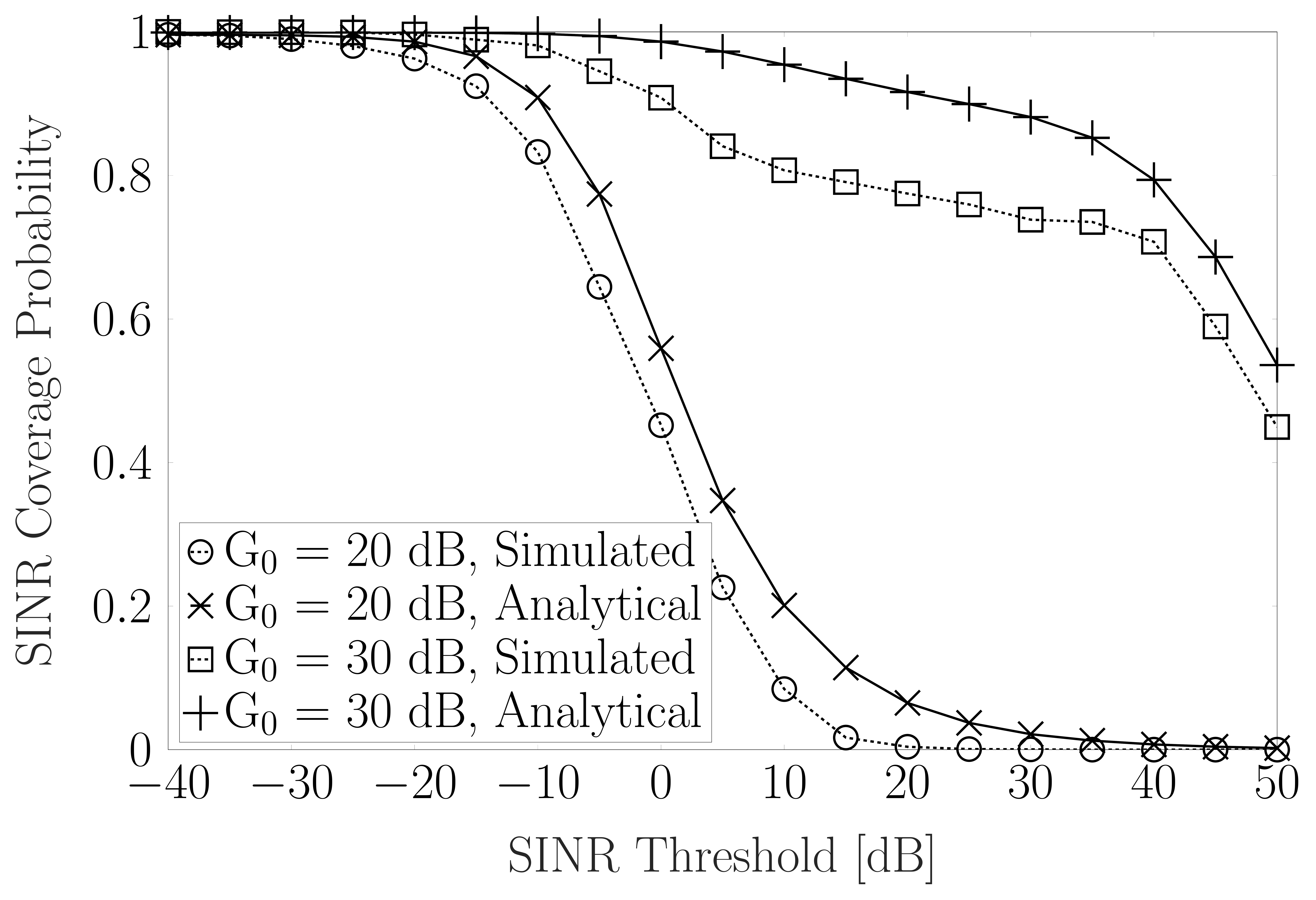}
\caption{Validation of the analytical model for SINR coverage probability, $\lambda_S = 0.1$~m$^{-1}, \lambda_R = 1e-5$~m$^{-2}$.}
\label{fig:SINR_valid}
\end{figure}

 \begin{figure}[t]
\centering
\includegraphics[width=8cm,height = 4cm]{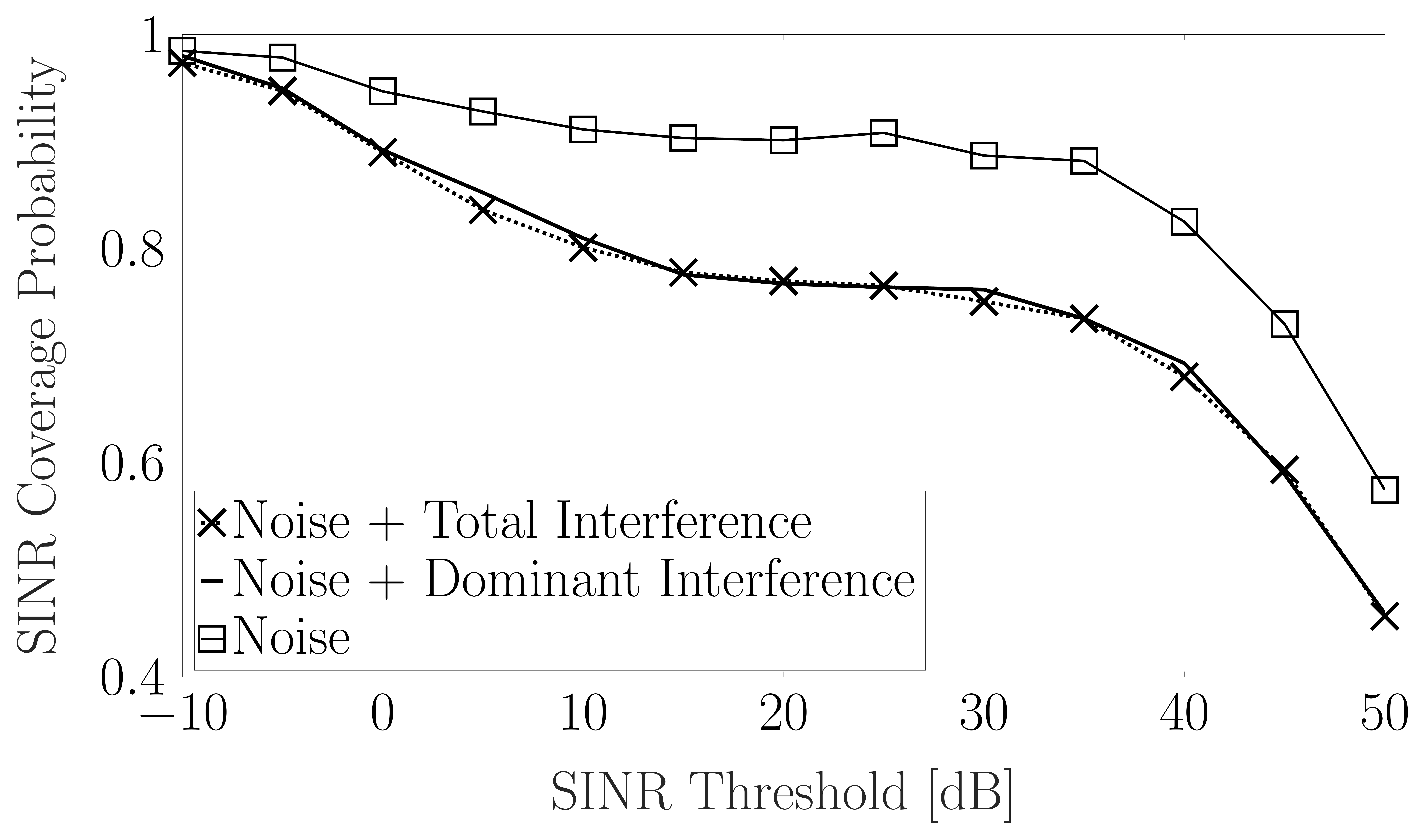}
\caption{Validation of the mm-wave interference model, $\lambda_S = 0.1$~m$^{-1}$, $G_0 = 30$ dB.}
\label{fig:Dominant_valid}
\end{figure}
The last integral of \eqref{eq:pdf_ds1} does not have a closed form. Consequently, we simplify the evaluation by expanding the exponential term in the numerator, i.e., $\exp(-2\lambda_S\sqrt{x^2-r^2})$, with a power series, and evaluating each of the resulting integral terms separately. Furthermore, we use Newton–Cotes quadrature rule \cite{abramowitz1964handbook} to evaluate the outer integral of \eqref{eq:PGF_COX}, as obtaining a closed form is not straightforward.
To validate these approximations, in Fig. \ref{fig:SINR_valid}, we compare the SINR coverage probability obtained using our analytical framework with Monte Carlo simulations. We observe that the analytical results agree appreciably with the simulations. Specifically, we observe that the analytical results provide a tight upper bound to the simulations.

Furthermore, we also validate our assumption of the dominant interferer model to characterize the mm-wave interference (Section~\ref{Sec:Prel} A). In Fig.~\ref{fig:Dominant_valid}, we use Monte Carlo simulations to compare the actual SINR coverage probability of the typical user with that obtained by considering the interference only from dominant user, and the one considering a noise limited model. We see that the noise limited model is not a true representation of the actual SINR characteristics, whereas, the dominant interferer model quite accurately matches with the actual SINR coverage probability. Thus, the dominant interferer model can be used to represent the mm-wave interference.

    \label{Sec:NRD}
    \subsection{Association and RAT Selection Probabilities}
    For the typical user, the perceived SBS density depends on both $\lambda_R$ and $\lambda_S$. However, the effects of $\lambda_R$ and $\lambda_S$ are quite different. In Fig. \ref{fig:Asso_SBS} we plot $\mathbb{P}_{SL}$ and $\mathbb{P}_{SN}$. As $\lambda_S$ increases for a given $\lambda_R$, the LOS SBS association probability increases. This is due to the fact that with increasing $\lambda_S$, the distance to the nearest SBS decreases. Although the number of NLOS SBSs also increases with increasing $\lambda_S$, their proximity to the typical user do not necessarily decrease significantly due to the fixed $\lambda_R$. 
  On the contrary, with increasing $\lambda_R$, with increasing $\lambda_R$, we observe that $\mathbb{P}_{SN}$ increases (see Fig. \ref{fig:Asso_road}. This is due to the decreasing proximity of NLOS SBSs with increasing $\lambda_R$.

In Fig. \ref{fig:RAT} we plot the conditional mm-wave selection probability with respect to $\lambda_S$, given that the typical user has selected a LOS SBS. We observe that increasing $G_0$ has a more pronounced effect on the mm-wave RAT selection than increasing $\lambda_S$. For $G_0 =  26$ dB, $\lambda_S =  20$/km$^2$ ensures mm-wave service. Whereas, with 25 dB, the operator needs to have $\lambda_S = 100$/km$^2$ (5 fold increase). Thus, increasing the antenna gains in the transmitter and/or receiver is a more effective way of prioritizing mm-wave selection, than deploying more SBSs.
    
\begin{figure}[t]
\centering
\includegraphics[width=8cm,height = 4cm]{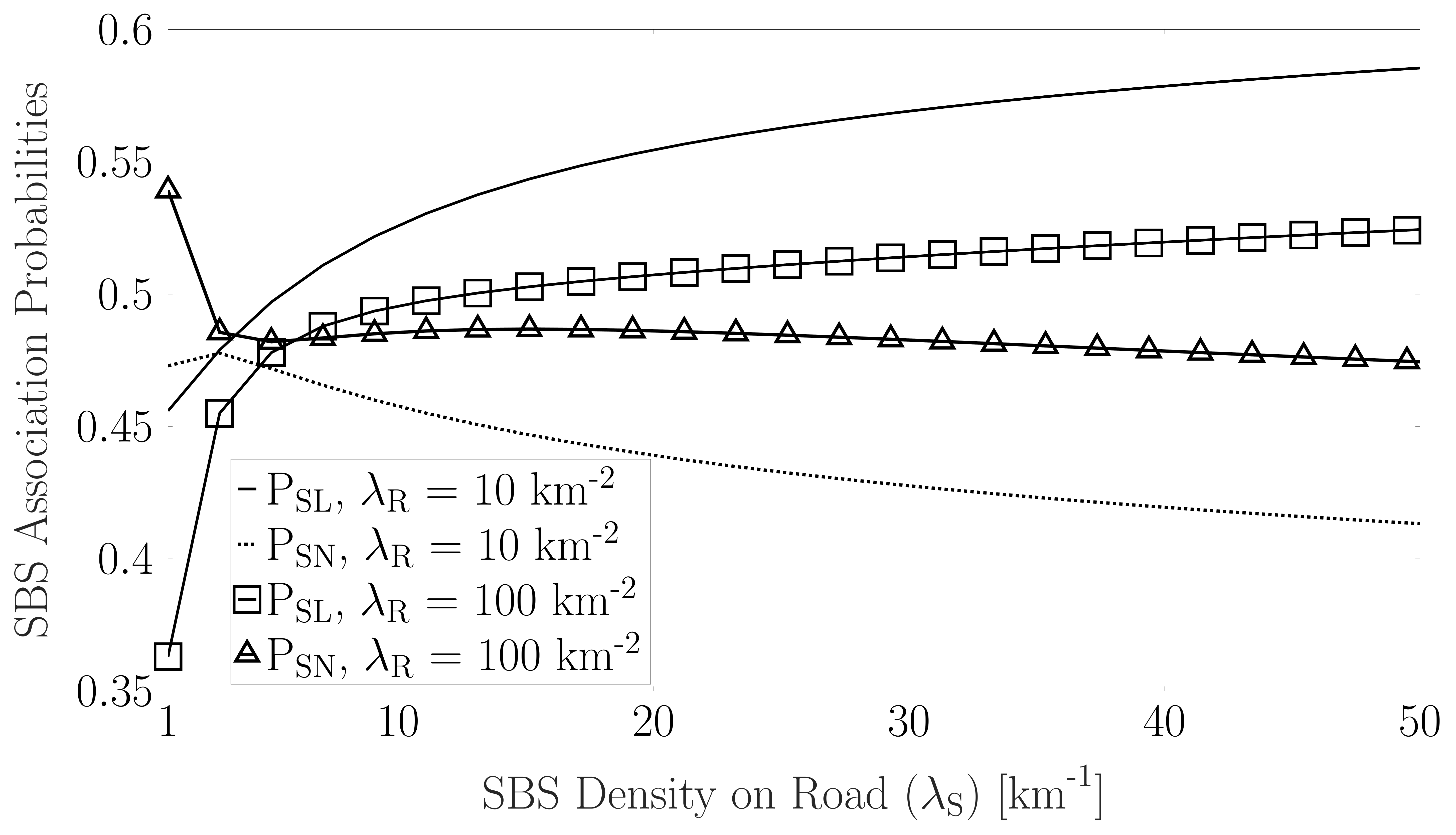}
\caption{Association probabilities vs SBS density for different road density.}
\label{fig:Asso_SBS}
\end{figure}
\begin{figure}[t]
\centering
\includegraphics[width=8cm,height = 4cm]{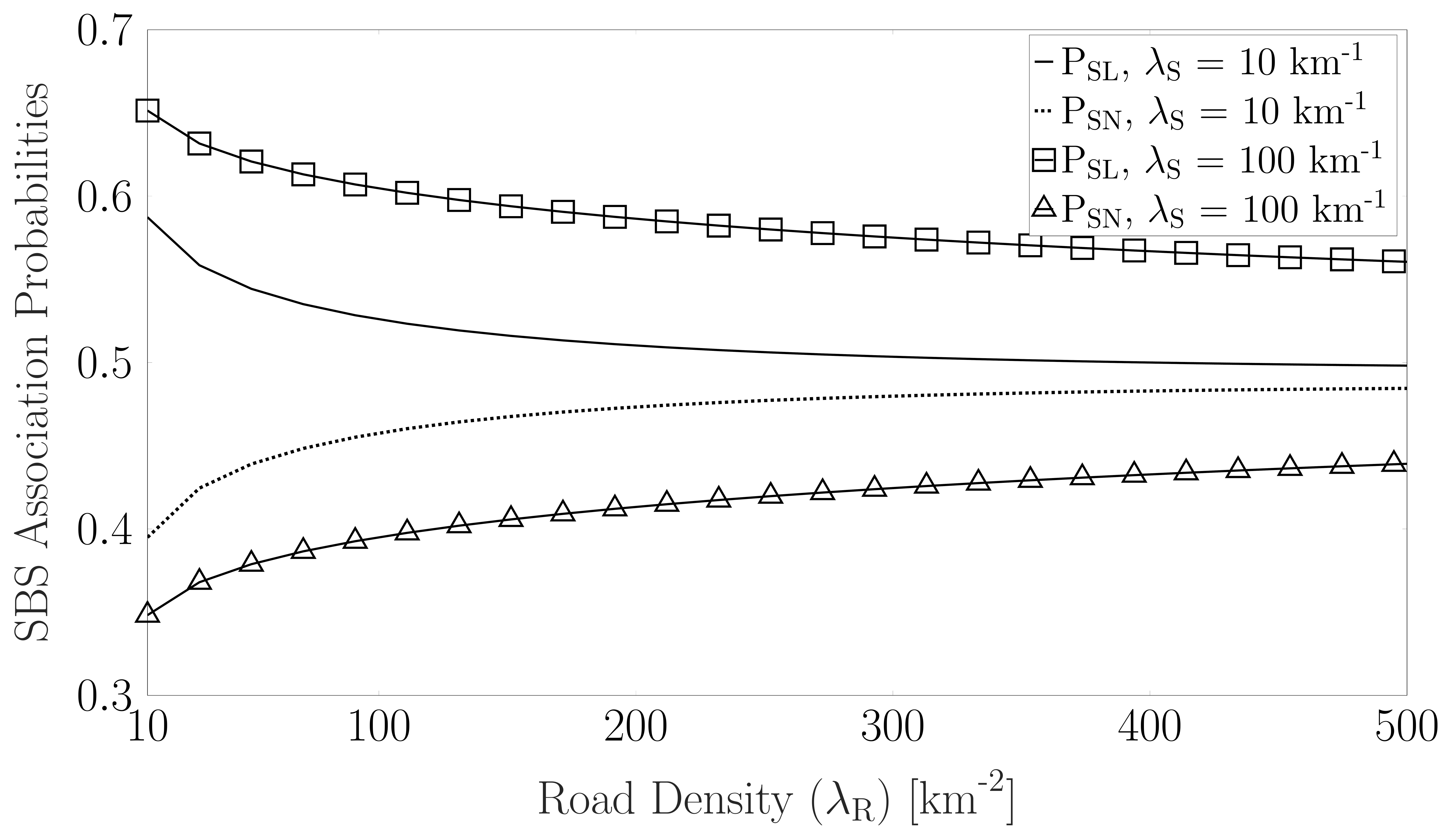}
\caption{Association probabilities vs road density for different SBS density.}
\label{fig:Asso_road}
\end{figure}
\begin{figure}[t]
\centering
\includegraphics[width=8cm,height = 4cm]{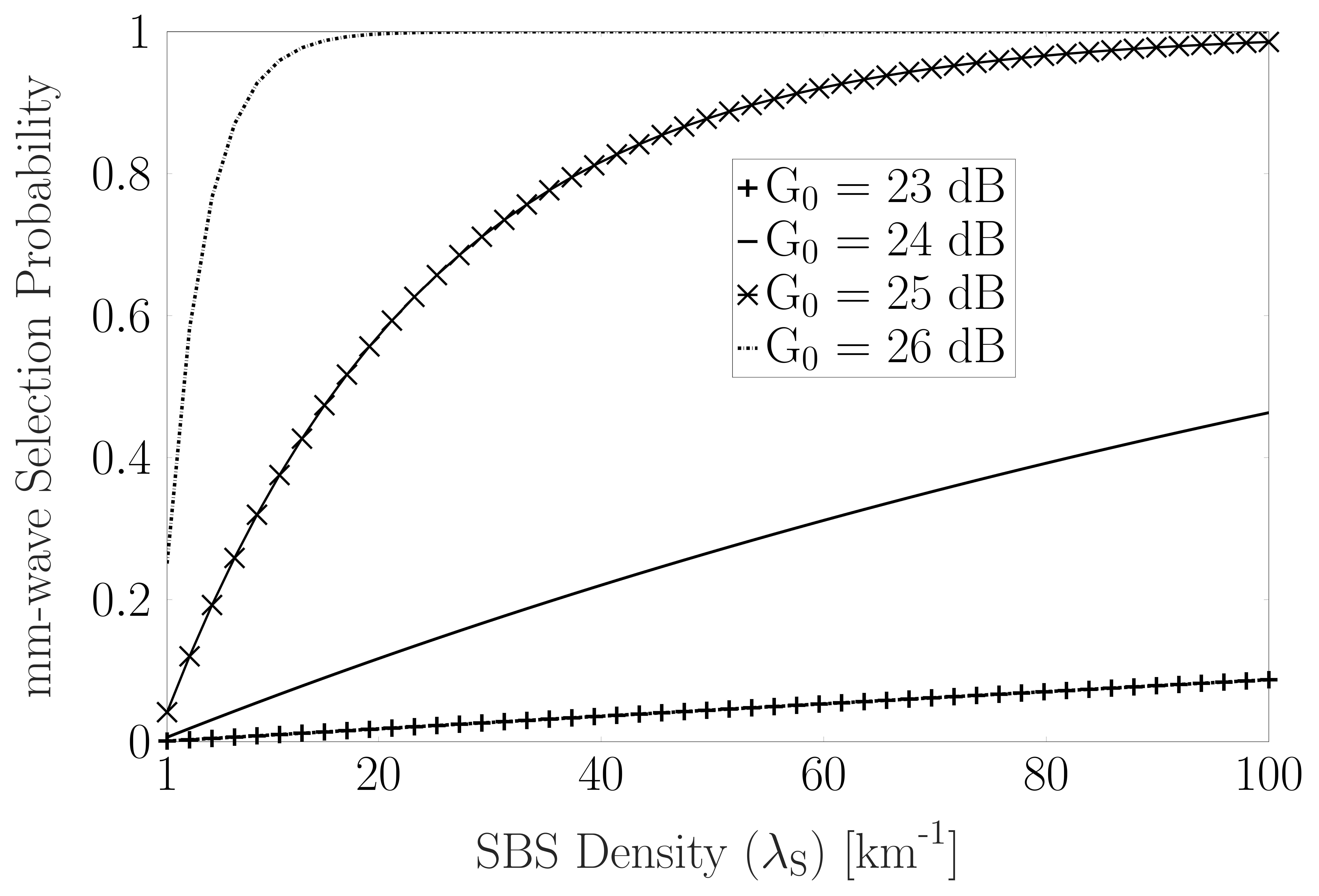}
\caption{Conditional mm-wave selection probability.}
\label{fig:RAT}
\end{figure}  
    
\subsection{SINR Coverage Probabilities}
     \begin{figure}[t]
\centering
\includegraphics[width=8cm,height = 4cm]{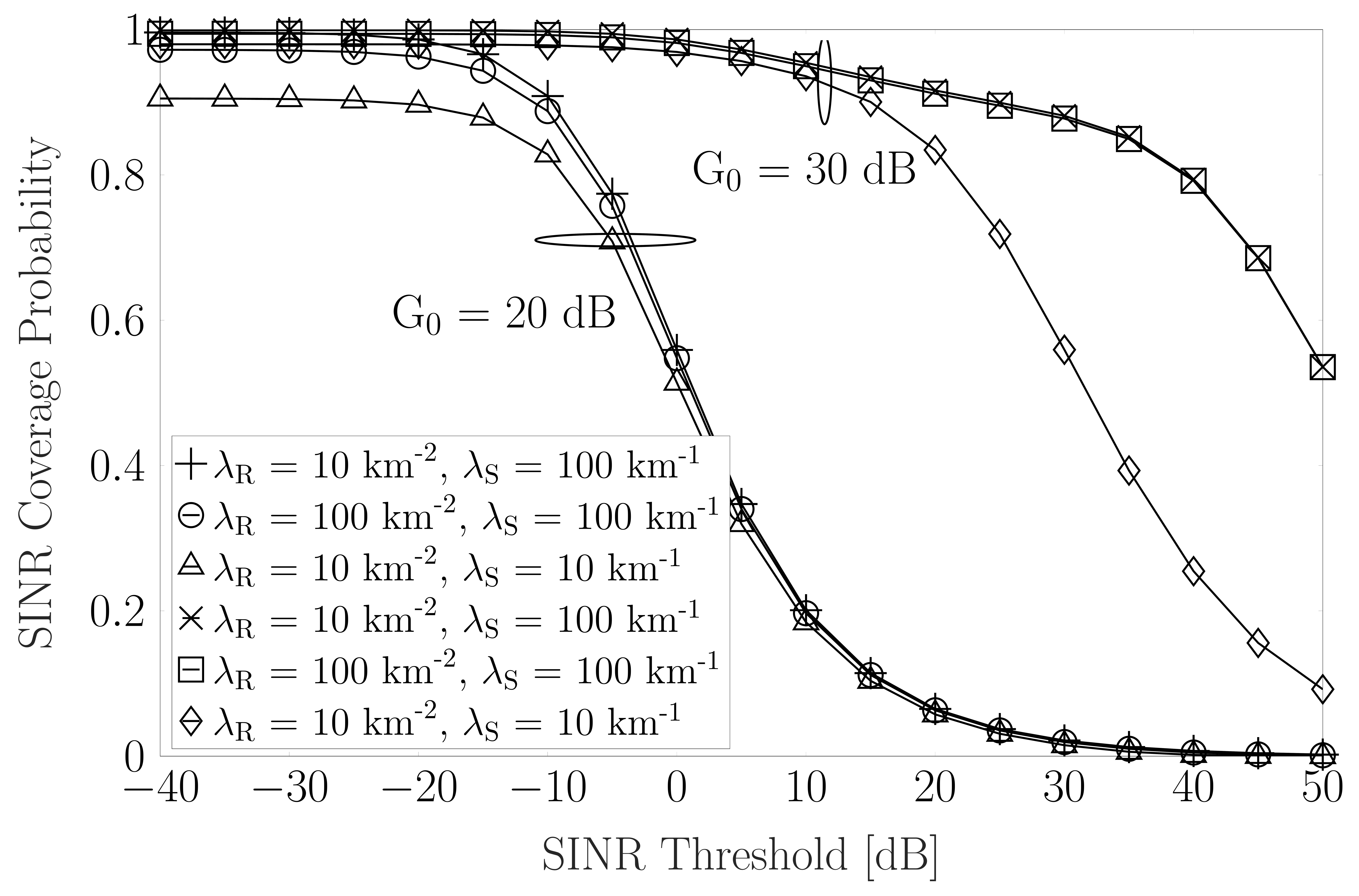}
\caption{SINR coverage probability for various road and SBS densities.}
\label{fig:SINR}
\end{figure}
\begin{figure}[t]
\centering
\includegraphics[width=8cm,height = 4cm]{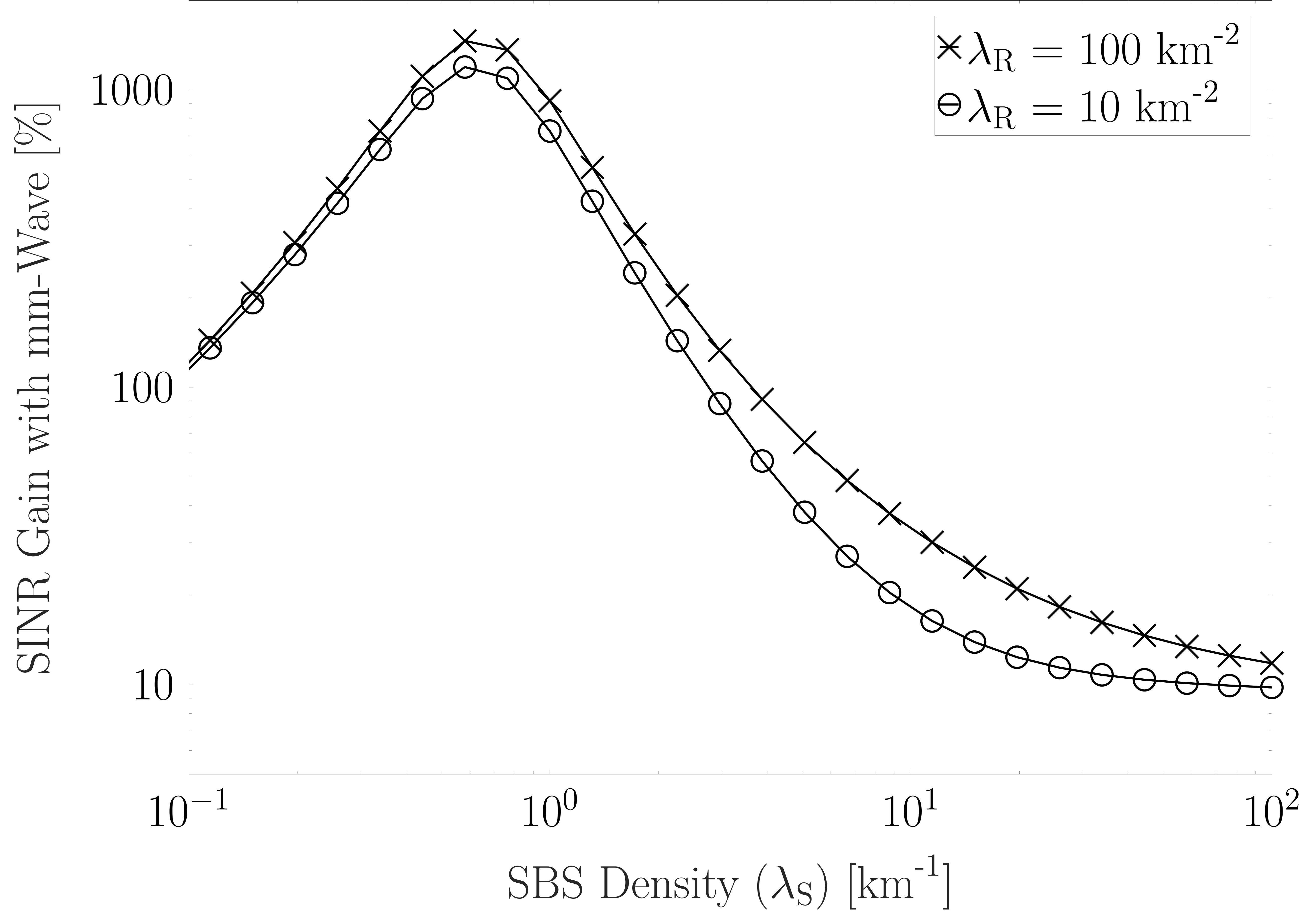}
\caption{SINR coverage probability gain with mm-wave, $\gamma$ = -10 dB.}
\label{fig:mm_vs_mu}
\end{figure}
    In Fig. \ref{fig:SINR} we plot the SINR coverage probability for different $\lambda_R$ and $\lambda_S$ and two different values of $G_0$. 
    Clearly, mm-wave (with $G_0 = 30$ dB) provides better SINR performance, precisely due to the large directional antenna gain and the fact that mm-wave transmissions suffer from minimal interference, i.e., only from the neighboring SBS.
    Furthermore, we observe that increasing $\lambda_R$ (i.e., going from a sparser to denser urban scenario), or decreasing $\lambda_S$, decreases the SINR performance of the user. The decrease in coverage with increasing $\lambda_R$ is because the interfering NLOS $\mu$-wave signals increase. This highlights the fact that, although with increasing road density, the number of SBSs perceived by the typical user increases, it does not necessarily improve the SINR performance of the user. Therefore, in denser urban scenarios, the operator should necessarily deploy more SBS per road, to maintain the SINR performance of the user.
    
    On the other hand, decreasing $\lambda_S$ increases the distance of the user from the nearest LOS SBS, thereby decreasing the useful signal power. This loss is more pronounced in the mm-wave association case with $G_0 = 30$ dB. This is due to the high path-loss of mm-wave signals, leading to severe deterioration in the useful signal power with decreasing $\lambda_S$.
    
Finally, we emphasize that the gain in SINR coverage by using mm-wave is dependent on $\lambda_S$. In Fig.~\ref{fig:mm_vs_mu}, we plot the gain with $G_0 = 30$ dB with respect to $G_0 = 20$ dB, at $\gamma = -10$ dB. With $G_0 = 20$ dB, the typical user mostly selects $\mu$-wave RAT, in contrast to mm-wave with $G_0 = 30$ dB. With $\lambda_S$, the gain initially increases, due to the decreasing proximity of the serving SBS. However, after a certain SBS density, the gain decreases due to increasing neighboring SBS interference. However, with very dense deployment we see that the gain saturates without decreasing further. 
Moreover, we see that with higher $\lambda_R$, the gain saturates at a higher value, as with higher $\lambda_R$, the $\mu$-wave performance deteriorates due to increasing NLOS SBS interference.
In fact, for very low $\lambda_S$ (e.g., $ \lambda_S \leq 10^{-4}$ m$^{-1})$, the gain may become negative, i.e., higher $\mu$-wave RAT selection would provide better SINR performance. However, such sparse SBS deployments may not be realistic in urban heterogeneous networks. Thus, the SBS density to maximize the SINR performance can be optimized, which we will study in a future work.
    \section{Conclusion}
    \label{Sec:Con}
We have analytically characterized a multi-tier heterogeneous network, where small cells are deployed along the roads and employ both $\mu$-wave and mm-wave RAT. We observed that going from a sparse to a more dense urban scenario, with more roads in a given region, does not necessarily increase the SINR performance of the user, even by keeping the SBS density per street constant. Increasing the SBS deployment in a street efficiently improves the  SINR coverage in $\mu$-wave operation. However, for mm-wave operation, too large SBS deployment leads to a saturation in the gain in SINR performance. In a future work we will investigate optimizing the deployment parameters to guarantee coverage, while taking downlink data rate into account.
    \bibliography{refer.bib}
\bibliographystyle{IEEEtran}
\appendices
\section{Proof of Lemma \ref{lem:p_G}}
\label{App:p_G}
We assume that a user is located at the center of the serving beam from its serving base station. Accordingly, the serving beam causes an interference region on the ground. We define "spillover" as the region of interference that the base stations create in their neighboring SBS while serving a user. The extent of this interference region depends on the height of the base stations, the beamwidth $\theta$ and the distance of the user from the base station.
Furthermore, we assume that the spillover region created by a base station while serving a user is limited to the neighboring SBS on the same side of the base station as that of the served user. Lastly, we assume that the spillover region does not extend beyond the neighboring SBS to the other side as that of the interfering SBS.

In what follows, we derive the probability that the typical user experiences mm-wave interference from the neighboring SBS. See Fig. \ref{fig:mmFig} for notations.

      Let the typical user $U_1$ be located at a distance $d_1$ from it's serving BS $B_1$ (the BS on the right in the figure). $U_1$ experiences mm-wave interference from the neighboring BS $B_2$ (the BS on the left in the figure), if it lies in the spillover region created by $B_2$, for some user $U_2$. We denote this spillover region by $s$.
     The probability that $U_1$ is located such that it falls in this spillover region is calculated as:
    \begin{align}
    \mathbb{P}\left(d_1  \geq \frac{d_a}{2} - s\right) = \exp\left(-\mu_{S}\left(\frac{d_a}{2} - s\right)\right) \label{eq:cond} ,
    \end{align}
where $d_a$ is the inter BS distance. This comes from the void probability of the PPP of the SBSs. Now, a user $U_2$, being served by the BS $B_2$, produces spillover to the coverage area of $B_1$ if and only if the extremest point of it's serving antenna beam crosses the cell boundary. In other words, the user $U_2$ produces spillover only if it's distance from $B_2$ is greater than some distance (say $d'$). Note that the maximum distance of $U_2$ from $B_2$ is $\frac{d_a}{2}$. Thus, to produce spillover in the coverage area of $B_1$, the user $U_2$ should lie in the region $d' \leq d_2 \leq \frac{d_a}{2}$.   
    The probability that at least one such user exists, and it's distance from it's serving SBS is between $d'$ and $\frac{d_a}{2}$ follows from the void probability and is obtained using the void probability of the user PPP and is given by $\left(1 - \exp\left(\mu_{OU}\left(\frac{d_a}{2}-d'\right)\right)\right)$.

For $U_2$, the spillover $(s)$ to the coverage area of $B_1$, caused due to $B_2$, can be calculated as:
    \begin{align}
    s &= GC - \frac{d_a}{2} = 
    h\tan\left(\theta + \psi \right) - \frac{d_a}{2} \nonumber \\
    &= h\tan\left(\theta + \left(\tan^{-1}\left(\frac{d_2}{h}\right) - \frac{\theta}{2}\right)\right)- \frac{d_a}{2} \nonumber \\ &= h\tan\left(\frac{\theta}{2} + \tan^{-1}\frac{d_2}{h}\right)- \frac{d_a}{2},
    \end{align}
where $\psi$ is the angle of depression from the top of $B_2$ to the nearest point of the serving beam of $U_2$ on the ground. Now, $d'$ is then obtained from the condition $s = 0$, i.e., the location of $U_2$, beyond which the coverage area of $B_1$ experiences spillover from $B_2$. This results in:
 \begin{align}
   d_2 = h\tan\left(\tan^{-1}\frac{d_a}{2h} - \frac{\theta}{2}\right)  = d' \nonumber 
    \end{align}
    
    Continuing our analysis, we impose the condition that no user on the left of $B_2$ effects in a spillover in the coverage region of $B_1$. Thus we have:
    \begin{align}
    d' \geq 0 \implies d_a \geq 2h\tan\left(\frac{\theta}{2}\right)
    \label{eq:cond1}
    \end{align}
  Lastly, we have the condition that $s$ cannot go beyond $B_1$, i.e., $s \leq \frac{d_a}{2}$. This holds true for all positions of $U_2$, specifically, at its maximum value i.e., $\frac{d_a}{2}$. This results in :
\begin{align}
&h\tan\left(\frac{\theta}{2} + \arctan\left(\frac{d_a}{2h}\right)  \right) \leq d_a \nonumber \\
\implies & \tan\left(\frac{\theta}{2}\right) \leq \frac{d_a h}{2h^2 + x^2}\nonumber \\
\implies & \tan\left(\frac{\theta}{2}\right) d_a^2 - h d_a + 2h^2\tan\left(\frac{\theta}{2}\right) \leq 0\nonumber \\
\implies  &\beta_1 \leq d_a \leq \beta_2 = \hat{d}
\label{eq:cond2}
\end{align}
where, 
\begin{align}
\beta_1 = \frac{h + \sqrt{h^2 - 8h^2\tan\left(\frac{\theta}{2}\right)}}{2\tan\left(\frac{\theta}{2}\right)} \\
\beta_2 = \frac{h + \sqrt{h^2 - 8h^2\tan\left(\frac{\theta}{2}\right)}}{2\tan\left(\frac{\theta}{2}\right)}
\end{align}
  Thus, from \eqref{eq:cond1} and \eqref{eq:cond2},
 \begin{align}
 d_a \geq \max \left(\beta_1,2h\tan\left(\frac{\theta}{2}\right)\right) = d^*
 \end{align} 
  
    Now we substitute this value of $s$ in \eqref{eq:cond}, and take the expectation with respect to $d_a$ and $d_2$. The joint distribution of $d_a$ and $d_2$ can be obtained according to the following reasoning. Assume that the random variables $d_a$ and $d_2$ are represented as: $d_a = X$ and $d_2 = Y$.   Now,
    \begin{align}
    f_{X,Y}(x,y) &= f_{X|Y}(x|y) f_Y(y) \nonumber \\
    &= \frac{-\delta}{\delta x}\mathbb{P}\left(X < x| Y = y\right)\frac{-\delta}{\delta y}\mathbb{P}\left(Y < y\right)\nonumber \\
    & \stackrel{a}{=}\frac{-\delta}{\delta x} \left[\exp\left(-\lambda_S(x-2y)\right)\right] \frac{-\delta}{\delta y} \left[\exp\left(-2\lambda_S y\right)\right] \nonumber \\
    &= \left( \lambda_S \exp(-\lambda_S (x-2y)) \right) \cdot \left(2\lambda_S \exp(-2\lambda_S y)\right) \nonumber \\
    & = 2\lambda_S^2 \exp(-\lambda_S (x)),
    \label{eq:pdf_distxy}
    \end{align}
where the step (a), the conditional probability is evaluated by the  following reasoning: given the fact that $B_2$ is located at a distance $y$ on any side of the user on the line, we calculate the probability of another base station (here $B_1$) on the other side of the user, at a distance greater than $y$ from the user, i.e., at a distance greater than $2y$ from $B_2$.

\section{Proof of Lemma \ref{lem:nearpt}}
\label{App:nearpt}
Consider that the nearest point of the NLOS SBS process from the typical user is at a distance $x$. Accordingly, the ball $\mathcal{B}(o,x)$ does not contain any NLOS SBS. We know that the number of lines of the Poisson line process hitting $\mathcal{B}(0,x)$ is Poisson distributed with parameter $2\pi\lambda_R x$~\cite{chiu2013stochastic}. Now, a randomly orientated line at a distance $r$ from the origin, has a chord length of $2\sqrt{x^2 - r^2}$, and a void probability $\exp(-2\lambda_S\sqrt{x^2 - r^2})$, in the circle $\mathcal{B}(0,x)$. As a result, the probability of no points falling in this ball, averaged over the number of lines, is calculated as:
   \begin{align}
   F_{d_{S1}}(x) &= \sum_{n=0}^{\infty} \frac{\left(2\pi\lambda_S x\right)^n\exp\left(-2\pi\lambda_S x\right)}{n!\, (x^{n})}\nonumber \\ &\left[\int_{r_1,r_2, ..., r_n = 0}^x\prod\limits_{i=1}^n\exp\left(-2\mu\sqrt{x^2 - r_i^2}\right)dr_i\right],  \nonumber
   \\& =  \sum_{n=0}^{\infty} \underbrace{\frac{\left(2\pi\lambda_S x\right)^n\exp\left(-2\pi\lambda_S x\right)}{n!\, (x^{n})}}_{A_1}\nonumber\\&\left[\underbrace{\int_{0}^x\exp\left(-2\mu\sqrt{x^2 - r^2}\right)dr}_{A_2}\right]^n, 
   \end{align}
   where the contribution from each of the chords is taken in the Palm sense, i.e, we calculate the void probabilities conditioned on the distances $r_i$ where we evaluate the integral in the range $0 \leq r_i \leq x$,  followed by dividing the integral by the Borel measure of the range i.e., $x$ for each chord. This results in the term $x^n$ in the denominator. The second term comes from the symmetry of the process $\phi_S$, i.e, contribution of each of the chords is equivalent on an average. The PDF of the distance $x$ is calculated by differentiating $F_{d_{S1}}(x)$ with respect to $x$:
   \begin{align}
   f_{d_{S1}}(x) &= -\frac{dF(x)}{dx} = -\sum_{n} \left[\frac{dA_1}{dx}A_2^n
 + A_1\frac{dA_2^n}{dx}\right] \nonumber \\ 
 & = \sum_{n} \left[\frac{(2\pi\lambda_S)^{n+1}}{n!} \exp(-2\pi\lambda_Sx)A_2^n + \right.\nonumber \\  & \left. \frac{(2\pi\lambda_S)^{n}}{n!} \exp(-2\pi\lambda_Sx)(nA_2^{n-1}A_3)\right], \nonumber \\
 & = 2\pi\lambda_S\exp(-2\pi\lambda (x-A_2)) \left[1+A_3\right] \nonumber
   \end{align}

\begin{align}
\mbox{where,    }   A_3 &=-1+ 2\mu x\int_0^x \frac{\exp(-2\mu\sqrt{x^2-r^2})}{\sqrt{x^2 - r^2}}dr \nonumber 
\end{align}
   \section{Proof of Lemma \ref{lem:Cox_PGF}}
   \label{App:Cox}
   The expression for PGF can be derived similarly to the derivation expression of the Laplace functional in \cite{6260478}. We start with a bounded support for $\nu(x)$, i.e. a disk centered at origin with radius $R$, and for the general case, the result follows from the monotone convergence theorem with increasing $R$.
    \begin{align}
       & G_{\phi_S}(\nu) = \mathbb{E}\left[\prod_{x\in\phi_S}\nu(x)\right] = \int\prod_{x\in\phi_S}\nu(x)\phi_S(dx) \nonumber \\
        &= \sum_0^{\infty} \frac{\exp\left(-2\pi R\lambda_S\right)}{n!\,(R)^{n}}\left(2\pi R\lambda_S\right)^n \nonumber \\ & \int_{r_1,r_2,\ldots,r_n = 0}^{R} \left(\prod_{i=1}^{n}\int_{\mathbb{R}}\prod_{x\in\phi_i}\nu(x)\phi_S(dx)\right) dr_1,\ldots,dr_n \nonumber 
    \end{align} 
    Now,
\begin{align}
&\quad\prod_{x\in\phi_i}\nu(x)\phi_S(dx) = \nonumber \\ &\exp\left(-2\mu_S\int_{0}^{\sqrt{R^2-r^2}} 1 - \nu\left(\sqrt{r_i^2 + t^2}\right)dt\right) \nonumber
\end{align}
As a result, we have:
\begin{align}
&\quad G_{\phi_S}(\nu) = \sum_0^{\infty} \frac{\exp(-2\pi R\lambda_S)\left(2\pi\lambda_S\right)^n}{n!} \nonumber \\ & \left(\int_{0}^{R}
\exp\left(-2\mu_S\int_{0}^{\sqrt{R^2-r^2}} 1 - \nu\left(\sqrt{r^2 + t^2}\right)dt\right)dr\right)^n \nonumber
\end{align}

   \section{Proof of Lemma \ref{lem:asso}}
   \label{App:asso}
   The probability of association with a LOS and NLOS MBS are given by:
\begin{align}
\mathbb{P}_{ML} &= \mathbb{E}[\mathds{1}(ML)]\mathbb{P}(P_{ML1} \geq P_{SL1})\mathbb{P}(P_{ML1} \geq P_{SN1})  \nonumber \\ \mathbb{P}_{MN} & =  \nonumber \left(1 - \mathbb{E}[\mathds{1}(ML)]\right) \nonumber \\ \nonumber &\mathbb{P}(P_{MN1} \geq P_{SL1})\mathbb{P}(P_{MN1} \geq P_{SN1}).
\end{align}
Here the term $\mathbb{P}_{ML}$, is a product of the probabilities of the existence of at least one LOS MBS, the probability that the received power from this strongest LOS MBS is larger than that received from the strongest LOS SBS and the strongest NLOS SBS. The term $\mathbb{P}_{MN}$ is developed similarly.

In the following we show calculate the terms of $\mathbb{P}_{MN}$. The terms for $\mathbb{P}_{ML}$ follows similarly. We have,

\begin{align}
 &\mathbb{P}(P_{MN1} \geq P_{SL1})  = \mathbb{P}\left(K_\mu P_Md_{M1}^{-\alpha_{MN}} \geq K_\mu P_Sd_{S1}^{-\alpha_{SL}} \right) \nonumber \\
 & = \mathbb{P}\left(d_{M1} \leq \left(\frac{P_S}{P_M}\right)^{-\frac{1}{\alpha_{MN}}}d_{SL1}^{\frac{\alpha_{MN}}{\alpha_{SL}}}\right) \nonumber \\
 & =\mathbb{E}_{d_{SL1}}\left[1 - \exp\left(-\pi\lambda_M\left(\frac{P_S}{P_M}\right)^{-\frac{2}{\alpha_{MN}}} d_{S1}^{\frac{2\alpha_{SN\mu}}{\alpha_{MN}}}\right)\right]. \nonumber \\
 &= 2\lambda_S \int_0^\infty \left(1 - \exp\left(-\pi\lambda_M\left(\frac{P_S}{P_M}\right)^{-\frac{2}{\alpha_{MN}}} x^{\frac{2\alpha_{SN\mu}}{\alpha_{MN}}}\right)\right) \nonumber \\ &\exp\left(-2\lambda_Sx\right) \nonumber 
\end{align}
Similarly, we can obtain $\mathbb{P}(P_{ML1} > P_{SL1})$.
On the other hand,
 \begin{align}
   &\mathbb{P}(P_{MN1} \geq P_{SN1})  = \nonumber \\  &\mathbb{E}_{d_{SN1}}\left[1 - \exp\left(-\pi\lambda_M\left(\frac{P_S}{P_M}\right)^{-\frac{2}{\alpha_{MN}}} d_{S1}^{\frac{2\alpha_{SN\mu}}{\alpha_{MN}}}\right)\right], \nonumber
   \end{align}
   where the expectation is taken with respect to the pdf of the first point of the NLOS SBS process. In the same way, we can obtain $\mathbb{P}(P_{ML1} \geq P_{SN1})$.
   
Now for the LOS SBS process we have:
\begin{align}
\mathbb{P}_{SL} = \mathbb{P}(P_{SL1} > P_{SN1})\left(\mathbb{P}(P_{SL1} > P_{ML1}) \mathbb{E}[\mathds{1}(ML)] + \right. \nonumber \\ \left.  \mathbb{P}(P_{SL1} > P_{MN1}) \left(1 - \mathbb{E}[\mathds{1}(ML)]\right)\right) \nonumber \end{align}

Here the first term corresponds to the probability that the received power from the strongest LOS SBS $(P_{SL1})$ is greater than that received from the strongest NLOS SBS. This is then multiplied by the probabilities $P_{SL1}$ is greater than the power received from the strongest LOS MBS, in case an LOS MBS exists, otherwise we consider the probability that $P_{SL1}$ is greater than the power received from the strongest NLOS MBS.
We have:
\begin{align}
& \mathbb{P}(P_{SL1} > P_{SN1}) 
 = \mathbb{E}_{d_{SN1}}\left[1-\exp\left(-2\mu d_{SN1}^{\frac{\alpha_{SN\mu}}{\alpha_{SL\mu}}}\right)\right] \nonumber
\end{align} is calculated using the void probability of the LOS SBS process.
For the MBSs, we have $\mathbb{P}(P_{SL1} > P_{Mv1}) = 1 - \mathbb{P}(P_{Mv1} > P_{SL1})$, for $v \in \{L,M\}$.
The association probability with the NLOS SBS tier can be calculated as: $\mathbb{P}_{SN} = 1 - \mathbb{P}_{ML} - \mathbb{P}_{MN} - \mathbb{P}_{SL}$.
   
   \section{Proof of Theorem \ref{theo:SINR}}
   \label{App:SINR}
   The derivations for the SINR coverage probability in the $\mu$-wave association case is fairly straightforward, and can be found in \cite{bai2015coverage,elshaer2016downlink,andrews2011tractable} etc. We present the proof sketch for one association case. The other cases follow on similar lines. In case the user is associated to a NLOS SBS, we have:
\begin{align}
SINR_{SN\mu} = \frac{P_S K_\mu h_{SN1} d_{SN1}^{-\alpha_{SN}}}{\sigma^2_{\mu} + I_{SN} + I_{ML} + I_{SL} + I_{MN}} \nonumber
\end{align}
\begin{align}
&\mathbb{P}\left(SINR_{SN\mu} \geq \gamma \right) \nonumber \\ & = \mathbb{P}\left(h_{SN1}  > \frac{\gamma\left(\sigma^2_{\mu} + I_{SN} + I_{ML} + I_{SL} + I_{MN}\right)}{P_S K_\mu  d_{SN1}^{-\alpha_{SN\mu}}}\right) \nonumber
\end{align}
where, $I_{(.)}$ are the interference terms from different tiers.
The expression is evaluated by using the tail distribution of the exponentially distributed $h_{SN1}$, followed by the independence of the different BS process. 
We provide the steps for evaluation of the term corresponding to the MBS LOS process. 
The other terms are obtained similarly.
\begin{align}
\mathbb{E}&\left[\exp  \left(-\frac{\gamma \sum\limits_{\phi_{ML}} I_{ML}}{P_S K_\mu  d_{SN1}^{-\alpha_{SN\mu}}}\right)\right]  \nonumber \\
& = \mathbb{E}\left[\exp\left(-\frac{\gamma P_M K_{\mu}\sum\limits_{\phi_{ML}} h_{MLi}d^{-\alpha_{ML}}_{MLi}}{P_S K_\mu  d_{SN1}^{-\alpha_{SN\mu}}}\right)\right] \nonumber \\
&\stackrel{a}{=}\mathbb{E}\left[\prod\limits_{\phi_{ML}} \exp\left(-\frac{\gamma  P_M h_{MLi}d^{-\alpha_{ML}}_{MLi}}{P_S  d_{SN1}^{-\alpha_{SN\mu}}}\right)\right] \nonumber \\
&\stackrel{b}{=}\mathbb{E}\left[\prod\limits_{\phi_{ML}} \frac{P_Sd_{MLi}^{\alpha_{ML}}}{P_Sd_{ML}^{\alpha_{ML}} + \gamma P_Md_{SN1}^{\alpha_{SN}}}\right] \nonumber \\
&=\mathbb{E}_{d_{SN1}}\left[G_{\phi_M}\left(\frac{P_Sx^{\alpha_{ML}}}{P_Sx^{\alpha_{ML}} + \gamma P_Md_{SN1}^{\alpha_{SN}}}\right)\right] \nonumber,
\end{align}
The step (a) follows from the independence of the variables $h_{MLi}$, (b) is obtained by applying the Laplace functional of $h_{MLi}$. Moreover, as per Lemma~\ref{lem:COX}, in case the user is associated to an NLOS SBS, the interfering NLOS SBS process $\phi_{SN})$ consists of the stationary $\phi_S$ and a line process $\phi_i$, passing though the serving SBS. Accordingly, the SINR coverage probability for NLOS SBS association has an additional term, which takes the line process into account.

For the mm-wave association case, we consider the interference only from the neighboring SBS. 
    Accordingly, we have:
\begin{align}
&\mathbb{P}\left(SINR_{SLm} \geq \gamma \right) = \nonumber \\  &\mathbb{P}\left(h_{SL1} \geq \frac{\gamma\sigma^2_{mm} +\gamma K_mP_S h_{SL2}d_{SL2}^{-\alpha_{SLm}}p_GG_0}{K_mP_Sd_{SL1}^{-\alpha_{SLm}}G_0}\right) \nonumber
\end{align}
where $d_{SL2}$ is the distance of the neighboring SBS.
Using Alzer's lemma for the tail distribution of a gamma random
variable with integer parameter \cite{alzer1997some}, Lemma 1, and from the definition of the PGF, the result \eqref{eq:SINR2} follows. The expectation is taken with respect to the distances of the serving and the neighboring SBS for the typical user.
Let the distance of the typical user from the serving and the neighboring SBS be given by $x$ and $y$ respectively. Thus the inter SBS distance between the serving and the interfering SBS is $x+y$. Now we calculate the joint distribution of $x$ and $y$ similar to that derived in \eqref{eq:pdf_distxy}:
 \begin{align}
    f_{X,Y}(x,y) &= 2\lambda_S^2 \exp(-\lambda_S (x +  y)).
    \end{align}
Taking the expectation with respect to the above joint distribution completes the proof.

\end{document}